\documentclass[11pt]{article}
\usepackage[height=220mm,width=150mm]{geometry}
\usepackage{array}
\usepackage{color}
\usepackage{mathrsfs}
\usepackage{hyperref}
\usepackage{amssymb}
\usepackage{amsmath}
\usepackage{amsthm}
\usepackage{color}
\allowdisplaybreaks
\usepackage{graphicx}
\usepackage{subfigure}
\usepackage{tikz}
\usepackage[pagewise]{lineno}
\usetikzlibrary{shapes.arrows,chains,positioning}
\newtheorem{Theorem}{Theorem}[section]

\newtheorem{Definition}[Theorem]{Definition}
\newtheorem{Proposition}[Theorem]{Proposition}

\def\nocolor#1{}

\begin{document}
%
\title{Stable Recovery of Weighted Sparse Signals from Phaseless Measurements via Weighted $l_1$ Minimization\thanks{This work was partially supported by the National Natural Science Foundation of China(Grant Nos. 11801256 and 12061044).}}

\author{Haiye Huo\\
Department of Mathematics, School of Science, Nanchang University, \\Nanchang~330031, Jiangxi, China\\
\mbox{} \\
Email:  hyhuo@ncu.edu.cn}

\date{}
\maketitle

\textbf{Abstract:}\,\,
The goal of phaseless compressed sensing is to recover an unknown sparse or approximately sparse signal from the magnitude of its measurements. However, it does not take advantage of any support information of the original signal. Therefore, our main contribution in this paper is to extend the theoretical framework for phaseless compressed sensing to incorporate with prior knowledge of the support structure of the signal. Specifically,
we investigate two conditions that guarantee stable recovery of a weighted $k$-sparse signal via weighted $l_1$ minimization without any phase information. We first prove that the weighted null space property (WNSP) is a sufficient and necessary condition for the success of weighted $l_1$ minimization for weighted $k$-sparse phase retrievable. Moreover, we show that if a measurement matrix satisfies the strong weighted restricted isometry property (SWRIP), then the original signal can be stably recovered from the phaseless measurements.




\section{Introduction}\label{sec:S0}
Over the past few years, compressed sensing \cite{CT2005,Donoho2006,EK2012,LES2019} has attracted considerable attention in various fields including signal processing, optics, and information theory. Basically, the goal of compressed sensing is to recover a sparse signal of interest from a small number of (noisy) linear measurements via the $l_p$ minimization with $0<p\leq1$ \cite{CWX2010a,ML2011,Sun2010,ZKL2017}. It is known that the restricted isometry property (RIP) and the null space property (NSP) are two typical conditions on the measurement matrix such that the exact/stable recovery of sparse signals can be guaranteed \cite{ACP2012,CDD2009}.

The $l_1$ minimization has been most widely used for sparse signal recovery in the compressed sensing literature, because of its convexity and easy computation. However,
such traditional $l_1$ minimization does not incorporate any prior support information of the original signal, which might help improve overall performance. In many practical applications, the support information of a signal is usually available as a prior. By exploiting such partial support information, there has been a large amount of research on the recovery of a sparse signal via the weighted $l_1$ minimization \cite{BMP2008,CL2019,FMS2012,GC2019,Jacques2010,MDR2017,VL2010}.
Moreover, based on the weighted $l_1$ minimization, a weighted sparse recovery problem has been intensively studied in \cite{BW2016,Flinth2016,HSX2018,RW2016} as well, in which a weight function is considered into the sparsity structure. Specifically, given a weight function $\mathbf{w}=\{w_i\}_{i=1}^N$ with $w_i\ge 1$, a signal $\mathbf{x}\in \mathbb{C}^N$ is called a weighted $k$-sparse signal, if
\begin{equation}\label{W-Sp}
\|\mathbf{x}\|_{\mathbf{w},0}=\sum_{\{i:\;|x_i|>0\}}w_i^2\le k.
\end{equation}
Then, the recovery of a weighted $k$-sparse signal by using the weighted $l_1$ minimization is referred to as the weighted sparse recovery problem. For more details, we refer the reader to \cite{BW2016,Flinth2016,HSX2018,RW2016}.

Phase retrieval is a fundamental problem of recovering a signal only from the magnitude of its linear measurements. It can be found in several areas,  such as radar signal processing \cite{Jaming2014}, quantum mechanics \cite{LR2009}, and optics \cite{SEC2015,Walter1963}. In these areas of application, the original signal to be recovered is often sparse. Hence, it is very natural to combine phase retrieval with compressed sensing, which is called phaseless compressed sensing \cite{Gao2017,IVW2017,VX2016,YHW2017}. The goal of phaseless compressed sensing is to reconstruct an unknown sparse signal $\mathbf{x_0}\in \mathbb{C}^N$ from the magnitude of its noisy measurements $\mathbf{y}=|\mathbf{Ax_0}|+\mathbf{e}$, where $\mathbf{A}\in \mathbb{C}^{m\times N}$ is the measurement matrix, $\mathbf{e}$ is the noise vector, and $|\cdot|$ denotes the element-wise absolute value. If the measurement matrix $\mathbf{A}$ satisfies the strong RIP (SRIP) \cite{GWX2016,VX2016} or the NSP \cite{WX2014}, $\mathbf{x}$ can be stably recovered by the following $l_1$ minimization up to a global phase:
\begin{equation}\label{mod:l0}
\min_{\mathbf{x}}\|\mathbf{x}\|_{1} \quad{\mbox{subject to}}\quad\|\mathbf{|Ax|}-|\mathbf{Ax_0}|\|_{2}^2\le\epsilon^2.
\end{equation}
You et al \cite{YHW2017} considered the phaseless compressed sensing via its nonconvex relaxation-$l_p$ minimization with $0<p<1$, and obtained a constant $p^*>0$ such that for any $p\in(0,p^*)$, every optimal solution to the $l_p$ minimization solves the considered problem too.

Obviously, the above $l_1$ minimization (\ref{mod:l0}) does not take into account the support information of the signal. Therefore, Zhou et al \cite{ZY2017} presented that the SRIP and the weighted NSP (WNSP) are two conditions for the success of $k$-sparse signal recovery from phaseless compressed sensing measurements via the weighted $l_1$ minimization when partial support information is available as a prior.
Zhang et al \cite{ZMH2019} used another different way to study the problem of phaseless compressed sensing using partial support information. They proposed two concepts of the partial NSP (P-NSP) and the partial SRIP (P-SRIP); and proved that the P-NSP and the P-SRIP are two exact reconstruction conditions on measurements for the problem of partially sparse phase retrieval.
To the best of our knowledge, there is no study on recovery of a weighted $k$-sparse signal from the phaseless measurements via the weighted $l_1$ minimization. Therefore, to fill this gap, our main contribution of this paper is to build up the theoretical framework for recovery of a weighted $k$-sparse signal from the magnitude of its measurements. To be specific, the mathematical model of the weighted $l_1$ minimization is represented as follows:
\begin{equation}\label{model:wl}
\min_{\mathbf{x}}\|\mathbf{x}\|_{\mathbf{w},1},\quad \mbox{subject to}\quad \||\mathbf{Ax}|-|\mathbf{Ax_0}|\|_2^2\le\epsilon^2,
\end{equation}
where the weighted $l_1$-norm $\|\mathbf{x}\|_{\mathbf{w},1}$ is given by
\[
\|\mathbf{x}\|_{\mathbf{w},1}=\sum_{i=1}^{N}w_i|x_i|,
\]
and $\mathbf{w}=\{w_i\}_{i=1}^N$ is the weight function with $w_i\ge 1$.
In this paper, we first give a sufficient and necessary condition, i.e., the WNSP, which can guarantee the unique recovery of a weighted $k$-sparse signal up to a global phase. Moreover, for the noisy setting, we propose a new concept, called the strong weighted RIP (SWRIP), and it is proved to be another complementary sufficient condition for stable recovery.

The rest of this paper is organized as follows. In Section~\ref{sec:S2}, we introduce the WNSP and show that it is a sufficient and necessary condition for exactly reconstructing a weighted sparse signal in phaseless compressed sensing. In Section~\ref{sec:S3}, we propose a new concept, i.e., SWRIP, and prove that stable recovery of a weighted $k$-sparse signal can be guaranteed if the measurement matrix satisfies the SWRIP.  Finally, we conclude this paper in Section~\ref{sec:S6}.

\emph{Notations:}
Let $[1:m]=\{1,2,\cdots,m\}$, and $\mathbf{A}=(a_1,a_2,\cdots,a_m)^T\in \mathbb{C}^{m\times N}$ be the measurement matrix. The null space of $\textbf{A}$ is given by
\[
\mathcal{N}(\mathbf{A})=\{\mathbf{x}:\;\mathbf{Ax}=\mathbf{0}\}.
\]
For a vector $\mathbf{x}\in\mathbb{R}^N$, its entries are denoted as $x_i,\; 1\leq i\leq N$. The complement of a set $\mathcal{S}\subset\{1,2,\cdots,N\}$ is defined by $\mathcal{S}^{c}=\{1,2,\cdots,N\}\backslash\mathcal{S}$.
$\mathbf{x}_{\mathcal{S}}$ is denoted as the sub-vector of $\mathbf{x}$, whose entries only with indices in $\mathcal{S}$ are kept.
For a weight $\mathbf{w}\in\mathbb{R}^N$, the weighted cardinality of a set $\mathcal{S}\subset\{1,2,\cdots,N\}$ is {\nocolor{red}{denoted}} as $w(\mathcal{S})=\sum_{i\in\mathcal{S}}w_i^2$.
The best weighted $k$-term approximation error is defined as
\[
\sigma_{k}(\mathbf{x})_{\mathbf{w},1}:=\min_{\mathbf{z}\in\Sigma_{\mathbf{w},k}^{N}}\|\mathbf{x}-\mathbf{z}\|_{\mathbf{w},1},
\]
where
\[
\Sigma_{\mathbf{w},k}^{N}=\{\mathbf{x}\in \mathbb{C}^N:\;\|\mathbf{x}\|_{\mathbf{w},0}\le k\}.
\]

\section{The Weighted Null Space Property}\label{sec:S2}

In this section, for any weighted $k$-sparse signal $\mathbf{x_0}\in \sum_{\mathbf{w},k}^N$, we consider the weighted $l_1$ minimization (\ref{model:wl}) without noise:
\begin{equation}\label{model:NW}
\min_{\mathbf{x}}\|\mathbf{x}\|_{\mathbf{w},1},\quad \mbox{subject to}\quad |\mathbf{Ax}|=|\mathbf{Ax_0}|.
\end{equation}
Similar to phaseless compressed sensing for a sparse signal, we explore the WNSP condition for the success of the weighted $l_1$ minimization for weighted $k$-sparse phase retrievable.

\subsection{The Real Case}

We first consider the real case of the problem, i.e., the signal $\mathbf{x}_0$ of interest and its measurement matrix $\mathbf{A}$ are in the real number field. We show that the WNSP is a sufficient and necessary condition for unique recovery of a weighted $k$-sparse signal $\mathbf{x}_0$ from its phaseless measurements up to a global phase.

\begin{Theorem}\label{Thm:PRN}
Given a measurement matrix $\mathbf{A}\in \mathbb{R}^{m\times N}$, the following two statements are equivalent:
\begin{enumerate}
  \item[(a)] For any $x_0\in\Sigma_{\mathbf{w},k}^{N}$, we have
  \[
  \mathop{\arg\min}_{\mathbf{x}\in \mathbb{R}^{N}}\{\|\mathbf{x}\|_{\mathbf{w},1}: |\mathbf{Ax}|=|\mathbf{Ax_0}|\}=\{\pm\mathbf{x_0}\}.
  \]
  \item[(b)] For all $S\subseteq[1:N]$ with $w(S)\le k$, it holds
  \[
  \|\mathbf{u}+\mathbf{v}\|_{\mathbf{w},1}<\|\mathbf{u}-\mathbf{v}\|_{\mathbf{w},1}
  \]
  for all nonzero $\mathbf{u}\in \mathcal{N}(\mathbf{A}_S)$ and $\mathbf{v}\in \mathcal{N}(\mathbf{A}_{S^c})$ satisfying $\mathbf{u}+\mathbf{v}\in\Sigma_{\mathbf{w},k}^N$.
\end{enumerate}
\end{Theorem}

\begin{proof}
$(a)\Rightarrow(b)$: We suppose that the statement $(b)$ does not hold. Then, there exists a subset $S\subseteq[1:N]$ with $w(S)\le k$,
nonzero $\mathbf{u}\in \mathcal{N}(\mathbf{A}_S)$ and $\mathbf{v}\in \mathcal{N}(\mathbf{A}_{S^c})$ such that
$\mathbf{u}+\mathbf{v}\in\Sigma_{\mathbf{w},k}^N$, and
\[
\|\mathbf{u}+\mathbf{v}\|_{\mathbf{w},1}\ge\|\mathbf{u}-\mathbf{v}\|_{\mathbf{w},1}.
\]
Let $\mathbf{x}_0:=\mathbf{u}+\mathbf{v},\;\mathbf{\hat{x}}:=\mathbf{u}-\mathbf{v}$. Then, $\mathbf{\hat{x}}\neq\pm\mathbf{x}_0$, and
\begin{equation}\label{Thm:PRN1}
\|\mathbf{\hat{x}}\|_{\mathbf{w},1}\le\|\mathbf{x}_0\|_{\mathbf{w},1}.
\end{equation}
Let $a_j^T,\;j=1,2,\cdots,m$ be the rows of the measurement matrix $\mathbf{A}$. From the definitions of $\mathbf{x}_0$ and $\mathbf{\hat{x}}$, we know
\[
2\langle a_j,\mathbf{u}\rangle=\langle a_j,\mathbf{x}_0+\mathbf{\hat{x}}\rangle, \quad j=1,2,\cdots,m,
\]
and
\[
2\langle a_j,\mathbf{v}\rangle=\langle a_j,\mathbf{x}_0-\mathbf{\hat{x}}\rangle, \quad j=1,2,\cdots,m.
\]
Since $\mathbf{u}\in \mathcal{N}(\mathbf{A}_S)$, we have
\[
\langle a_j,\mathbf{x}_0\rangle=-\langle a_j,\mathbf{\hat{x}}\rangle,\;\; {\rm{for}}\;\; j\in S.
\]
Similarly, we get
\[
\langle a_j,\mathbf{x}_0\rangle=\langle a_j,\mathbf{\hat{x}}\rangle,\;\; {\rm{for}}\;\; j\in S^c.
\]
Thus, we obtain
\[
|\mathbf{Ax}_0|=|\mathbf{A\hat{x}}|.
\]
It follows from (\ref{Thm:PRN1}) that $\mathbf{\hat{x}}$ is a solution to (\ref{model:NW}), which contradicts $(a)$.

$(b)\Rightarrow(a)$: We suppose that $(a)$ does not hold. Then, there exists a solution $\mathbf{\hat{x}}\ne\pm\mathbf{x}_0$ to (\ref{model:NW}), i.e.,
\begin{equation}\label{Thm:PRN2}
|\mathbf{A\hat{x}}|=|\mathbf{Ax_0}|,
\end{equation}
and
\begin{equation}\label{Thm:PRN3}
\|\mathbf{\hat{x}}\|_{\mathbf{w},1}\le\|\mathbf{x}_0\|_{\mathbf{w},1}.
\end{equation}
Let $a_j^T,\;j=1,2,\cdots,m$ be the rows of the measurement matrix $\mathbf{A}$. By (\ref{Thm:PRN2}), we know that there exists a subset $S\subseteq[1:m]$ satisfying
\begin{equation}\label{Thm:PRN4}
\langle a_j,\mathbf{x}_0+\mathbf{\hat{x}}\rangle=0,\;\; {\rm{for}}\;\; j\in S,
\end{equation}
and
\begin{equation}\label{Thm:PRN5}
\langle a_j,\mathbf{x}_0-\mathbf{\hat{x}}\rangle=0,\;\; {\rm{for}}\;\; j\in S^c.
\end{equation}
Let $\mathbf{u}:=\mathbf{x}_0+\mathbf{\hat{x}}$, and $\mathbf{v}:=\mathbf{x}_0-\mathbf{\hat{x}}$. Note that $\mathbf{\hat{x}}\ne\pm\mathbf{x}_0$, then combining (\ref{Thm:PRN4}) and (\ref{Thm:PRN5}), we have $\mathbf{u}\in \mathcal{N}(\mathbf{A}_{S})\backslash \{0\}$,\;$\mathbf{v}\in \mathcal{N}(\mathbf{A}_{S^c})\backslash \{0\}$, and $\mathbf{u}+\mathbf{v}=2\mathbf{x}_0\in\Sigma_{\mathbf{w},k}^N$. From $(b)$, we get
\[
\|\mathbf{u}+\mathbf{v}\|_{\mathbf{w},1}<\|\mathbf{u}-\mathbf{v}\|_{\mathbf{w},1},
\]
i.e.,
\[
\|\mathbf{x}_0\|_{\mathbf{w},1}<\|\mathbf{\hat{x}}\|_{\mathbf{w},1},
\]
which contradicts (\ref{Thm:PRN3}). This completes the proof.
\end{proof}

\subsection{The Complex Case}
Next, we consider the same problem in for the complex case. We say that $\mathcal{S}=\{S_1,S_2,\cdots,S_p\}$ is a partition of $[1:m]$, if
\[
S_j\subseteq[1:m], \; \bigcup_{j=1}^pS_j=[1:m],\; S_j\cap S_l=\emptyset \quad{\mbox{for all}}\quad j\ne l.
\]
Let $\mathbb{S}=\{c\in \mathbb{C}:|c|=1\}$. The next theorem is an extension of Theorem~\ref{Thm:PRN}.

\begin{Theorem}\label{Thm:Comp}
Given a measurement matrix $\mathbf{A}\in \mathbb{C}^{m\times N}$, the following two statements are equivalent:
\begin{itemize}
  \item [(a)] For any $\mathbf{x}_0\in \Sigma_{\mathbf{w},k}^{N}$, we have
  \begin{equation}\label{Thm:C0}
  \mathop{\arg\min}_{\mathbf{x}\in \mathbb{C}^N}\{\|\mathbf{x}\|_{\mathbf{w},1}:|\mathbf{Ax}|=|\mathbf{Ax}_0|\}=\{c\mathbf{x}_0:\;c\in \mathbb{S}\}.
  \end{equation}
  \item [(b)] Assume that $S=\{S_1,S_2,\cdots,S_p\}$ is any partition of $[1:m]$, and that $\eta_{j}\in \mathcal{N}(\mathbf{A}_{S_j})\backslash\{0\}$ with
  \begin{equation}\label{Thm:C1}
  \frac{\eta_1-\eta_l}{c_1-c_l}=\frac{\eta_1-\eta_j}{c_1-c_j}\in \Sigma_{\mathbf{w},k}^{N}\backslash\{0\}\quad \mbox{for all}\quad l,\;j\in[2:p],
  \end{equation}
  for some pairwise distinct $c_1,c_2,\cdots,c_p\in \mathbb{S}$. Then, we have
  \begin{equation}\label{Thm:C2}
  \|\eta_j-\eta_l\|_{\mathbf{w},1}<\|c_l\eta_j-c_j\eta_l\|_{\mathbf{w},1},
  \end{equation}
  for all $j,\;l\in[1:p]$ with $j\ne l$.
\end{itemize}
\end{Theorem}


\begin{proof}

The proof is similar to those in \cite[Theorem~3.3]{WX2014} and \cite[Theorem~3.2]{Gao2017}.

$(b)\Rightarrow(a)$: Suppose that the statement $(a)$ does not hold, then there exists a solution $\mathbf{\bar{x}}\notin\{c\mathbf{x_0}:\, c\in \mathbb{S}\}$ to (\ref{model:NW}), which satisfies
\begin{equation}\label{Thm:C3}
\|\mathbf{\bar{x}}\|_{\mathbf{w},1}\le\|\mathbf{x}_0\|_{\mathbf{w},1},
\end{equation}
and
\begin{equation}\label{Thm:C4}
|\mathbf{A\bar{x}}|=|\mathbf{Ax}_0|.
\end{equation}
Let $a_j^T,\; j=1,2,\cdots,m$ be the rows of the measurement matrix $\mathbf{A}$. By (\ref{Thm:C4}), we have
\begin{equation}\label{Thm:C5}
\langle a_j,\mathbf{\bar{x}}\rangle=\langle a_j,c_j\mathbf{x}_0\rangle,
\end{equation}
where $c_j\in \mathbb{S},\,j=1,2,\cdots,m$. We define an equivalence relation on $[1:m]$ by $\tilde{c}_j$, that is $j\sim l$, when $c_j=c_l$.
Then, the equivalence relation leads to a partition $\mathcal{S}=\{S_1,S_2\cdots,S_p\}$ of $[1:m]$. Let $c_j:=\tilde{c}_l,\;l\in S_j$. Obviously, $c_j, 1\le j\le p$ are distinct and belong to $\mathbb{S}$.  For any $S_j,\;j=1,2,\cdots,p$, we get
\[
\mathbf{A}_{S_j}\mathbf{\bar{x}}=\mathbf{A}_{S_j}(c_j\mathbf{x}_0).
\]
Let
\begin{equation}\label{Thm:C6}
\eta_j:=c_j\mathbf{x}_0-\mathbf{\bar{x}},\;\; j=1,2,\cdots,p.
\end{equation}
Then, we obtain $\eta_j\in \mathcal{N}(\mathbf{A}_{S_j})\backslash\{0\}$, and
\[
\frac{\eta_1-\eta_l}{c_1-c_l}=\frac{\eta_1-\eta_j}{c_1-c_j}=\mathbf{x}_0\in \Sigma_{\mathbf{w},k}^{N},\quad \mbox{for all}\;\; l,\;j\in[2:p],\; l\ne j.
\]
Hence, from statement (b), we know that
\begin{equation}\label{Thm:C7}
\|\eta_j-\eta_l\|_{\mathbf{w},1}<\|c_l\eta_j-c_j\eta_l\|_{\mathbf{w},1}.
\end{equation}
Substituting (\ref{Thm:C6}) into (\ref{Thm:C7}), we obtain
\[
\|(c_j-c_l)\mathbf{x}_0\|_{\mathbf{w},1}<\|(c_j-c_l)\mathbf{\bar{x}}\|_{\mathbf{w},1},
\]
i.e.,
\[
\|\mathbf{x}_0\|_{\mathbf{w},1}<\|\mathbf{\bar{x}}\|_{\mathbf{w},1},
\]
which is a contradiction with (\ref{Thm:C3}). Hence, the statement (a) holds.

$(a)\Rightarrow(b)$: Suppose that the statement $(b)$ does not hold, then there exists a partition $S=\{S_1,S_2,\cdots,S_p\}$ of $[1:m]$,
 $\eta_j\in \mathcal{N}(\mathbf{A}_{S_j})\backslash\{0\},\;j\in [1:p]$, and some pairwise distinct $c_1,c_2,\cdots,c_p\in \mathbb{S}$ satisfying (\ref{Thm:C1}), and
\begin{equation}\label{Thm:C8}
\|\eta_{j_0}-\eta_{l_0}\|_{\mathbf{w},1}\ge\|c_l\eta_{j_0}-c_j\eta_{l_0}\|_{\mathbf{w},1}
\end{equation}
for some distinct $j_0,\;l_0\in[1:p]$.
Let
\begin{equation}\label{Thm:C10}
\mathbf{x_0}:=\eta_{j_0}-\eta_{l_0}\in\Sigma_{\mathbf{w},k}^N,
\end{equation}
and
\begin{equation}\label{Thm:C9}
\mathbf{\hat{x}}:=c_{l_0}\eta_{j_0}-c_{j_0}\eta_{l_0},\; c_{l_0}\ne c_{j_0}.
\end{equation}
Then, we get
\[
\mathbf{\hat{x}}\notin \{c\mathbf{x}_0, \;c\in \mathbb{S}\},
\]
and
\begin{equation}\label{Thm:add}
\|\mathbf{\hat{x}}\|_{\mathbf{w},1}\le\|\mathbf{x}_0\|_{\mathbf{w},1}.
\end{equation}
Let $a_j^T,\;j=1,2,\cdots,m$ be the rows of the measurement matrix $\mathbf{A}$. Since $\eta_j\in \mathcal{N}(\mathbf{A}_{S_j})\backslash\{0\}$, we have
\[
\langle a_k,\eta_{j_0}\rangle=0\;\; \mbox{or} \,\;\langle a_k,\eta_{l_0}\rangle=0,\;\; k\in S_{l_0}\cup S_{j_0}.
\]
From the definitions of $\mathbf{x}_0$ and $\mathbf{\hat{x}}$ (see (\ref{Thm:C10}) and (\ref{Thm:C9})), we get
\begin{equation}\label{Thm:C11}
|\langle a_k,\mathbf{x}_0\rangle|=|\langle a_k,\mathbf{\hat{x}}\rangle|,\;\; k\in S_{l_0}\cup S_{j_0}.
\end{equation}
For any $k\notin S_{l_0}\cup S_{j_0}$, without loss of generality, we assume that $k\in S_{t},\; t\ne j_0,\;l_0$. Thus, $\langle a_k,\eta_{t}\rangle=0$.
By (\ref{Thm:C1}), we have
\begin{equation}\label{Thm:C12}
\frac{\eta_j-\eta_l}{c_j-c_l}=\frac{\eta_m-\eta_n}{c_m-c_n}\in \Sigma_{\mathbf{w},k}^{N}\backslash\{0\},
\end{equation}
where $j,l,m,n$ are distinct integers.
Let
\[
\mathbf{y}_0:=\frac{\eta_{j_0}-\eta_t}{c_{j_0}-c_t}=\frac{\eta_{l_0}-\eta_t}{c_{l_0}-c_t}.
\]
Thus,
\[
\eta_{j_0}=(c_{j_0}-c_t)\mathbf{y}_0+\eta_t,
\]
and
\[
\eta_{l_0}=(c_{l_0}-c_t)\mathbf{y}_0+\eta_t.
\]
Hence, $\mathbf{x}_0$ and $\mathbf{\hat{x}}$ can be rewritten as
\[
\mathbf{x}_0=(c_{j_0}-c_{l_0})\mathbf{y}_0,
\]
and
\[
\mathbf{\hat{x}}=(c_{j_0}-c_{l_0})c_t\mathbf{y}_0+(c_{l_0}-c_{j_0})\eta_t,
\]
respectively.
Since $\langle a_k,\eta_t\rangle=0$,
we have
\[
|\langle a_k, \mathbf{\hat{x}}\rangle|=|\langle a_k,\mathbf{x}_0\rangle|,\; k\in S_{t}.
\]
Using a similar argument, we obtain
\[
|\langle a_k, \mathbf{\hat{x}}\rangle|=|\langle a_k,\mathbf{x}_0\rangle|,\;\;\mbox{for all}\;\;k.
\]
It follows from (\ref{Thm:add}) that $\mathbf{\hat{x}}\notin\{c\mathbf{x}_0:\;c\in \mathbb{S}\}$ is a solution to (\ref{model:wl}), which is a contradiction with the statement $(a)$. This completes the proof.
\end{proof}

By Theorems \ref{Thm:PRN} and \ref{Thm:Comp}, we know that if the measurement matrix satisfies the WNSP, then an unknown weighted sparse signal can be exactly recovered by solving the weighted $l_1$ minimization model (\ref{model:NW}) up to a global phase. However, it is very hard to check whether the measurement matrix satisfies the condition (i.e., statement $(b)$) in Theorems \ref{Thm:PRN} and \ref{Thm:Comp} or not. To this end, we present another property, called the SWRIP, in the following section, as an alternative way to guarantee the uniqueness of the weighted $l_1$ minimization model (\ref{model:NW}).

\section{The Strong Weighted Restricted Isometry Property}\label{sec:S3}

In phaseless compressed sensing, Gao et al \cite{GWX2016} presented that if the measurement matrix satisfies the SRIP, then the model (\ref{mod:l0}) provides a stable solution. In this section, we propose to generalize the SRIP to the weighted sparsity setting, and investigate the conditions under which the weighted $l_1$ minimization model (\ref{model:NW}) guarantees stable recovery of a weighted sparse signal. In this section, we only focus on studying signals and matrices that are in the real number field.

Before stating the main results, we recall the definition of the weighted RIP (WRIP), and the conditions for stable recovery in the weighted sparse setting for traditional compressed sensing, where the optimization model considered is
\begin{equation}\label{model:cs1}
\min_{\mathbf{x}}\|\mathbf{x}\|_{\mathbf{w},1},\quad \mbox{subject to}\quad \|\mathbf{Ax}-\mathbf{Ax_0}\|_{2}^2\le\epsilon^2.
\end{equation}
\begin{Definition}[WRIP]\cite[Definition~1.3]{RW2016}\label{Def:WRIP}
Given a weight $\mathbf{w}\in\mathbb{R}^{N}$, a matrix $\mathbf{A}\in \mathbb{R}^{m\times N}$ is said to satisfy the WRIP of order $k$ with constant $\delta_{\mathbf{w},k}\in (0,1)$, if
\begin{equation}\label{Def:WRIP:1}
(1-\delta_{\mathbf{w},k})\|\mathbf{x}\|_{2}^2\le\|\mathbf{Ax}\|_2^2\le(1+\delta_{\mathbf{w},k})\|\mathbf{x}\|_2^2
\end{equation}
holds for all weighted $k$-sparse vectors $\mathbf{x}\in \mathbb{R}^{N}$.
\end{Definition}

\begin{Proposition}\cite[Theorem~3.4]{HSX2018}\label{lem:WRIP}
Suppose that $\mathbf{A}\in \mathbb{R}^{m\times N}$ satisfies the WRIP of order $2k$ with constant
\begin{equation}\label{lem:WRIP:A1}
\delta_{\mathbf{w},2k}<\frac{1}{2\sqrt{2}+1}
\end{equation}
for	$k\ge 2\|\mathbf{w}\|_{\infty}^2$. Let $\mathbf{x}\in \mathbb{R}^N$, $\mathbf{y}=\mathbf{Ax}+\mathbf{e}$ with $\|\mathbf{e}\|_2\le\epsilon$,
and $\tilde{\mathbf{x}}$ be the solution to $(\ref{model:cs1})$. Then, we have
\begin{equation}\label{lem:WRIP:A2}
\|\mathbf{x}-\tilde{\mathbf{x}}\|_{2}\le c_1\epsilon+\frac{c_2}{\sqrt{k}}\sigma_{k}(\mathbf{x})_{\mathbf{w},1},
\end{equation}
where
\begin{eqnarray}
c_1&=&\frac{6\sqrt{1+\delta_{\mathbf{w},2k}}}{1-(1+2\sqrt{2})\delta_{\mathbf{w},2k}},\nonumber\\ c_2&=&\frac{4(1+(\sqrt{2}-1)\delta_{\mathbf{w},2k})}{1-(1+2\sqrt{2})\delta_{\mathbf{w},2k}}.\label{lem:WRIP:A3}
\end{eqnarray}
\end{Proposition}

In this paper, we derive a sufficient condition for stable recovery of a weighted sparse signal from its phaseless measurements. For any weighted $k$-sparse signal $\mathbf{x_0}$, we turn back to consider the weighted $l_1$ minimization in phaseless compressed sensing:
\begin{equation}\label{model:wl1}
\min_{\mathbf{x}\in \mathbb{R}^N}\|\mathbf{x}\|_{\mathbf{w},1},\quad \mbox{subject to}\quad \||\mathbf{Ax}|-|\mathbf{Ax_0}|\|_{2}^2\le\epsilon^2.
\end{equation}

In the following, we first propose the notion of the SWRIP, which is a combination of the WRIP and the SRIP.
\begin{Definition}[SWRIP]\label{Def:WSRIP}
For a weight $\mathbf{w}\in\mathbb{R}^{N}$, a matrix $\mathbf{A}\in \mathbb{R}^{m\times N}$ is said to satisfy the SWRIP of order $k$ with bounds $\theta_{\mathbf{w},-},\;\theta_{\mathbf{w},+}\in (0,2)$, if
\begin{equation}\label{Def:WSRIP:1}
\theta_{\mathbf{w},-}\|\mathbf{x}\|_{2}^2\le\min_{I\subseteq[1:m],\;|I|\geq m/2}\|\mathbf{A}_I\mathbf{x}\|_2^2\le\max_{I\subseteq[1:m],\;|I|\geq m/2}\|\mathbf{A}_{I}\mathbf{x}\|_2^2\le\theta_{\mathbf{w},+}\|\mathbf{x}\|_2^2
\end{equation}
holds for all weighted $k$-sparse $\mathbf{x}\in \mathbb{R}^{N}$.
\end{Definition}

Based on the SWRIP, we present a reconstruction error estimation via the weighted $l_1$ minimization (\ref{model:wl1}) as follows.

\begin{Theorem}\label{Thm:S-WRIP}
Let $\mathbf{x_0}\in \mathbb{R}^N$, and $\mathbf{y}=\mathbf{Ax_0}+\mathbf{e}$ with $\|\mathbf{e}\|_2\le\epsilon$. Suppose that $\mathbf{A}\in \mathbb{R}^{m\times N}$ satisfies the SWRIP of order $2k$ with constants $\theta_{\mathbf{w},-}\in(1-\frac{1}{2\sqrt{2}+1},1),\;\theta_{\mathbf{w},+}\in(1,1+\frac{1}{2\sqrt{2}+1})$ for $k\ge 2\|\mathbf{w}\|_{\infty}^2$. Then, any solution $\hat{\mathbf{x}}$ to the weighted $l_1$ minimization $(\ref{model:wl1})$ satisfies
\begin{equation}\label{S-WRIP:SW1}
\min{\{\|\mathbf{\hat{x}}-\mathbf{x_0}\|_2,\;\|\mathbf{\hat{x}}+\mathbf{x_0}\|_2\}}\le c_1\epsilon+c_2\frac{\sigma_{k}(\mathbf{x_0})_{\mathbf{w},1}}{\sqrt{k}},
\end{equation}
where $c_1$ and $c_2$ are defined the same as (\ref{lem:WRIP:A3}) in Proposition~\ref{lem:WRIP}.
\end{Theorem}

\begin{proof}
For any solution $\mathbf{\hat{x}}\in\mathbb{R}^N$ to $(\ref{model:wl1})$, it holds
\begin{equation}\label{S-WRIP:SW2}
\|\mathbf{\hat{x}}\|_{\mathbf{w},1}\le\|\mathbf{x_0}\|_{\mathbf{w},1},
\end{equation}
and
\begin{equation}\label{S-WRIP:SW3}
\||\mathbf{Ax}|-|\mathbf{Ax_0}|\|_{2}^2\le\epsilon^2.
\end{equation}
Let $a_j^T,\;j=1,2,\cdots,m$ be the rows of the measurement matrix $\mathbf{A}$. We divide the index set $\{1,2,\cdots,m\}$ into two subsets:
\[
T=\{j:{\rm{sign}}(\langle a_j,\mathbf{\hat{x}}\rangle)={\rm{sign}}(\langle a_j,\mathbf{x_0}\rangle)\},
\]
and
\[
T^c=\{j:{\rm{sign}}(\langle a_j,\mathbf{\hat{x}}\rangle)=-{\rm{sign}}(\langle a_j,\mathbf{x_0}\rangle)\}.
\]
Then, we know that either $|T|\ge m/2$ or $|T^c|\ge m/2$. First, we assume that $|T|\ge m/2$. It follows from (\ref{S-WRIP:SW3}) that
\begin{equation}\label{S-WRIP:SW4}
\|\mathbf{A}_{T}\mathbf{x}-\mathbf{A}_{T}\mathbf{x_0}\|_{2}^2+\|\mathbf{A}_{T^c}\mathbf{x}+\mathbf{A}_{T^c}\mathbf{x_0}\|_{2}^2\le\epsilon^2.
\end{equation}
By (\ref{S-WRIP:SW4}), we know that
\begin{equation}\label{S-WRIP:SW5}
\|\mathbf{A}_{T}\mathbf{\hat{x}}-\mathbf{A}_{T}\mathbf{x_0}\|_{2}^2\le\epsilon^2.
\end{equation}
Combining (\ref{S-WRIP:SW2}) and (\ref{S-WRIP:SW5}), we get
\begin{equation}\label{S-WRIP:SW6}
\mathbf{\hat{x}}\in\{ \mathbf{x}\in\mathbb{R}^N:\; \|\mathbf{x}\|_{\mathbf{w},1}\le\|\mathbf{x_0}\|_{\mathbf{w},1},\;\|\mathbf{A}_{T}\mathbf{x}-\mathbf{A}_{T}\mathbf{x_0}\|_{2}\le\epsilon.\}
\end{equation}
Since $\mathbf{A}$ satisfies the SWRIP of order $2k$, then we know that $\mathbf{A}_T$ satisfies the WRIP of order $2k$ with
\begin{equation}\label{S-WRIP:SW7}
\delta_{\mathbf{w},2k}\le\max\{1-\theta_{\mathbf{w},-},\;\theta_{\mathbf{w},+}-1\}<\frac{1}{2\sqrt{2}+1}.
\end{equation}
Combining (\ref{S-WRIP:SW6}), (\ref{S-WRIP:SW7}) and Proposition~\ref{lem:WRIP}, we get
\[
\|\mathbf{\hat{x}}-\mathbf{x_0}\|_2\le c_1\epsilon+c_2\frac{\sigma_{k}(\mathbf{x_0})_{\mathbf{w},1}}{\sqrt{k}},
\]
where $c_1$ and $c_2$ are defined the same as in Proposition~\ref{lem:WRIP}. Similarly, for the case $|T^c|\ge m/2$, we obtain
\[
\|\mathbf{\hat{x}}+\mathbf{x_0}\|_2\le c_1\epsilon+c_2\frac{\sigma_{k}(\mathbf{x_0})_{\mathbf{w},1}}{\sqrt{k}}.
\]
This completes the proof.
\end{proof}

By Theorem~\ref{Thm:S-WRIP}, we know that if the measurement matrix satisfies the SWRIP, then the weighted $l_1$ minimization (\ref{model:wl1}) can provide a stable solution.

\section{Conclusion}\label{sec:S6}
In this paper, we study the stable recovery of a weighted $k$-sparse signal from the magnitude of its measurements via the weighted $l_1$ minimization. First, we prove that the WNSP is a sufficient and necessary condition on the measurement matrix for exactly reconstructing a weighted $k$-sparse signal in phaseless compressed sensing. Moreover, we propose a new concept, called the SWRIP, and it is proved to be a sufficient condition for weighted $k$-sparse phase retrievable. In the future, we will focus on investigating some numerical algorithms for stable recovery of weighted sparse signals from its phaseless measurements.

\end{document}